\documentclass[12pt]{article}

\usepackage{amsthm,amsmath,natbib}
\usepackage[margin=1in]{geometry}
\usepackage[affil-it]{authblk}
\usepackage{graphicx,bm,mathtools,amsfonts,bbm,hyperref}
\numberwithin{equation}{section}
\theoremstyle{plain}
\newtheorem{thm}{Theorem}[section]
\newtheorem{lemma}[thm]{Lemma}
\newcommand{\R}{\mathbb{R}}
\newcommand{\paren}[1]{\left(#1\right)}
\renewcommand{\brack}[1]{\left[#1\right]}
\newcommand{\defeq}{\vcentcolon=}
\newcommand{\sumi}{\sum_{i=1}^n}
\newcommand{\bmat}{\begin{pmatrix}}
\newcommand{\emat}{\end{pmatrix}}
\def\bal#1\eal{\begin{align*}#1\end{align*}}
\newcommand{\argsup}{\operatornamewithlimits{argsup}}
\usepackage{algorithmicx}
\usepackage{algorithm}
\usepackage[noend]{algpseudocode}
\algnewcommand\algorithmicinput{\textbf{Input:}}
\algnewcommand\Input{\item[\algorithmicinput]}
\algnewcommand\algorithmicoutput{\textbf{Output:}}
\algnewcommand\Output{\item[\algorithmicoutput]}
\newcommand{\set}[1]{\left\{#1\right\}}
\newcommand\norm[1]{\left\lVert#1\right\rVert}

\newcommand{\beto}{\beta_1}

\newcommand{\betost}{\beta_1^*}
\newcommand{\psio}{\psi^{(1)}}
\newcommand{\tab}{\hspace*{2em}}

\makeatletter 
\def\@maketitle{ 
\newpage 
\null 
\vskip 2em 
\begin{center} 
\let \footnote \thanks 
{\Large \bfseries \@title \par } 
\vskip 1.5em 
{ \normalsize 
\lineskip .5em 
\begin{tabular}[t]{c} 
\@author 
\end{tabular}\par} 
\end{center} 
\par 
\vskip 1.5em} 
\makeatother

\begin{document}

\title{Modeling microbial abundances and dysbiosis with beta-binomial regression}

\author{Bryan D. Martin \thanks{\texttt{bmartin6@uw.edu}}}
\affil{Department of Statistics, University of Washington}
\author{Daniela Witten}
\affil{Departments of Statistics and Biostatistics, University of Washington}
\author{Amy D. Willis}
\affil{Department of Biostatistics, University of Washington}

\maketitle

\begin{abstract}
Using a sample from a population to estimate the proportion of the population with a certain category label is a broadly important problem.
In the context of microbiome studies, this problem arises when researchers wish to use a sample from a population of microbes to estimate the population proportion of a particular taxon, known as  the taxon's \emph{relative abundance}.
In this paper, we propose a beta-binomial model for this task.
Like existing models, our model allows for  a taxon's relative abundance to be associated with covariates of interest.
However, unlike existing models, our proposal also allows for the overdispersion in the taxon's counts to be associated with covariates of interest.
We exploit this model in order to propose tests not only for differential relative abundance, but also for differential variability.
The latter is particularly valuable in light of speculation that  \emph{dysbiosis}, the perturbation from a normal microbiome that can occur in certain disease conditions, may manifest as a loss of stability, or increase in variability, of the counts associated with each taxon.
We demonstrate the performance of our proposed model using a simulation study and an application to soil microbial data.
\end{abstract}

\section{Introduction}

Estimating the proportion of a population that belongs to a certain category---the relative abundance---is a problem spanning fields as broad as social science, population health, and ecology.
For example, researchers may be interested in estimating the proportion of low-income students who attend competitive higher-education institutions \citep{bastedo2011running}, child mortality rates in Sub-Saharan African regions \citep{mercer2015space}, or the proportion of diseased leaf tissue in coastal grasslands \citep{parker2015phylogenetic}.
In most of these settings, it is not possible to sample the entire population of interest, and it is necessary to estimate the true proportion based on a sample of individuals from the population. 
In this paper we consider the general problem of estimating the prevalence of a category within a population when the category labels of the observed individuals may be correlated.

While this problem is of broad interest, our method is particularly motivated by the ever-increasing number of studies of microbiomes.
A microbiome is the collection of microscopic organisms (microbes), along with their genes and metabolites, that inhabit an ecological niche \citep{poussin2018interrogating}.
Microbes live on and in the human body, and in fact, microbial cells may outnumber human cells \citep{sender2016revised}. 
Because of this, the relative abundance of a microbe---or a \textit{taxon}, which refers to a biological grouping of microbes---is a common marker of host or environmental health.
For example, the species \textit{G. vaginalis} has been found to correlate with symptomatic bacterial vaginosis \citep{callahan2017replication}; 
different genera of Cyanobacteria flourish in response to precipitation and irrigation run-off 
\citep{tromas2018niche}; and Parkinson's disease has been associated with reduced levels of the family \textit{Prevotellaceae} \citep{hill2017parkinson}. 
Accurate and precise estimation of microbial abundances is critical for  disease diagnosis and treatment \citep{diag2014nature, diag2014, diag20142, diag2015}.

A particularly challenging aspect of estimating microbial abundances is that the category labels of microbes are known to be correlated.  
Microbial communities are spatially organized, 
with a member of one taxon more likely to be observed close to the same taxon than close to a different taxon \citep{welch2016biogeography}.
In this paper we argue that a correlated-taxon model is a natural approach to estimating relative abundances in this setting.
It successfully explains the large number of unobserved taxa in many samples, as well as overdispersion in the abundance of observed taxa relative to models where the occurrences of individual microbes are uncorrelated.

An additional advantage of our method is that it provides a statistical framework for testing for \textit{dysbiosis}.
Dysbiosis describes a microbial imbalance, or a deviation from a healthy microbiome \citep{petersen2014defining, hooks2017dysbiosis}.
In particular, the term is often used to refer to a change in the \textit{stability} of a microbiome.
For example, inflammatory bowel disease (IBD) has been associated with increases in the variability of the gut microbiome \citep{halfvarson2017dynamics}, and the microbiomes of IBD patients are often referred to as dysbiotic \citep{tamboli2004dysbiosis}.
Unlike many methods for modeling relative abundances of microbial taxa, the method that we propose provides a natural framework for hypothesis testing for dysbiosis via the parameters of a heteroskedastic model for taxon abundances.
Specifically, we can test whether the variability in a taxon's counts is associated with some covariate of interest.

Our paper is laid out as follows. In Section~\ref{s:litrev}, we review several existing regression models for microbial abundances.
In Section~\ref{s:meth}, we propose our model, and discuss parameter estimation. 
We propose approaches for testing for differential abundance and differential variability in Section~\ref{s:hyp}. 
In Section~\ref{s:simstudy}, we show via simulation that our hypothesis testing framework is valid, even with small sample sizes. 
We apply our method to data from a soil microbiome study in Section~\ref{s:data}, and we close with a discussion of our method in Section~\ref{s:disc}.
Software for implementing our model and hypothesis testing procedures is available in the \texttt{R} package \texttt{corncob}, available at \texttt{github.com/bryandmartin/corncob}.

\section{Literature Review}\label{s:litrev}

Modeling of population proportions, or \textit{relative abundances}, has a long history in the statistical literature, and includes basic methods such as z-tests for proportions, and logistic regression.
However, modeling microbial abundance data brings with it a number of challenges.
For example, the dynamic nature of the microbiome commonly gives rise to a large number of microbial taxa that are only present in a small number of samples, but are highly abundant when present
\citep{digiulio2015temporal, dethlefsen2011incomplete}.
Some microbes may be so rare that they consistently evade detection or are observed at low abundances in all samples \citep{sogin2006microbial}. 
In addition, 
the number of taxa (typically on the order of thousands) is generally substantially less than the number of samples (typically less than one hundred). 
Finally, the number of counts that are observed in each sample may differ substantially, and thus the amount of information contained in each sample may differ.

Thus, we focus our literature review on models for microbial abundances. 
We broadly categorize these models into two approaches:  jointly modeling multiple taxa, and modeling each taxon individually. 
While our proposal pertains to the latter, both approaches are common and each has its advantages and disadvantages, which we now review.

Jointly modeling multiple taxa is a popular approach because it represents the entire microbial community with a single model. 
However, since these communities are often very diverse (the total number of taxa is large), and different taxa exhibit differing levels of variability,
a large number of parameters is typically needed to obtain a good model fit \citep{kurtz2015sparse, sankaran2017latent}. 
Hierarchical models of absolute abundances are often used to constrain the number of parameters (e.g. \cite{la2012hypothesis, holmes2012dirichlet, chen2013variable, sankaran2017latent, cao2017microbial}). 
However, modeling the variance structure is challenging with few parameters \citep{sankaran2017latent}.
Many joint taxon models make use of the log-ratio or centered log-ratio transformations to model relative abundances.
However, these approaches typically cannot be applied to zero-valued observations \citep{aitchison1986statistical, mcmurdie2014waste, willis2018divnet}.
Since many taxa are typically unobserved in each sample, these methods commonly make use of pseudo-counts to replace zeros, or incorporate a zero-inflation component into their model \citep{xia2013logistic, mandal2015analysis, MZILN, willis2018divnet}.
In the case of pseudo-counts, parameter estimation depends on an arbitrarily chosen hyperparameter, while zero-inflated models may lack interpretability. 

Because simultaneously modeling large numbers of microbial taxa is challenging, an alternative approach is to model individual taxa one-by-one.
We further classify individual taxon models into models for observed relative abundances (the proportion of the observed counts that corresponds to the specific taxon), and models for absolute abundances (the number of observed counts of the taxon). 
A particularly common model for observed relative abundances is the beta distribution, which is a natural choice since it is supported on $(0,1)$. 
Zero-inflated beta regression models have been proposed to account for the large number of zeros often observed in microbial abundance data, corresponding to the absence of a taxon in a sample \citep{ZIB, zibr, chai2018marginalized}.
Non-parametric models for observed relative abundances \citep{white2009statistical, segata2011metagenomic} and Gaussian models for transformed observed relative abundances \citep{morgan2012dysfunction, morgan2015associations} have also been proposed.

Another option is to model the absolute abundance of a taxon.
Popular methods originally designed for RNAseq data, such as DESeq2 \citep{love2014moderated} and EdgeR \citep{robinson2010edger}, make use of the negative binomial distribution. 
These models can be extended with random effects and a zero-inflation component to account for correlation across subjects and to model additional overdispersion of the counts \citep{NBMM, fang2016zero}.
Alternative approaches to modeling absolute abundances include the use of transformations such as cumulative sum scaling \citep{wahba1995smoothing, paulson2013differential}, trimmed mean of M-values \citep{robinson2010scaling, law2014voom}, and ratio approaches \citep{sohn2015robust, chen2018gmpr}.

All of the papers mentioned thus far focus on an association between \textit{mean} abundance and covariates. 
In this paper, we propose a beta-binomial regression model for microbial taxon abundances.
To the best of our knowledge,
this is the first regression model that allows for an association between the \textit{variance} of a taxon's abundance and covariates, rather than only an association between the mean abundance and covariates.
In addition, our model can accommodate the absence of a taxon in samples, variability in the total number of counts across samples, and high variability in the observed relative abundances.

\section{The Beta-Binomial Regression Model}\label{s:meth}

\subsection{A Hierarchical Model for Microbial Abundances} \label{ss:model}

In this section, we present a beta-binomial regression model for microbial abundance data.
While the beta-binomial model has been extensively studied in the statistics literature \citep{skellam1948probability, kleinman1973proportions, williams1975394, prentice1986binary, mccullagh1989generalized, aerts2002topics, dolzhenko2014using, wagner2015importance}, to our knowledge, we are the first to propose a regression framework that can link both discrete and continuous covariates to both a relative abundance parameter and a correlation/overdispersion parameter, as well as the first to apply this model to the analysis of microbial data.
We summarize the notation and definitions defined in this section in Table~\ref{tab:notation}.

\begin{table}[ht!]
\centering
\begin{tabular}{c|c}
\textbf{Notation} & \textbf{Definition}\\\hline
$Y_{i,j}$ & \textit{indicator that the $j^{th}$ read corresponds to the taxon of interest}\\
$W_i$ & {\it observed counts}, or {\it observed absolute abundance}, of the taxon of interest\\
$M_i$ & \textit{sequencing depth}, or \textit{total number of counts}, across all taxa\\
$W_i/M_i$ & \textit{observed relative abundance} of the taxon of interest\\
$Z_i$ & \textit{latent relative abundance} of the taxon of interest\\
$\mu_i$ & \textit{expected relative abundance} of the taxon of interest\\
$\phi_i$ & \textit{overdispersion}, or \textit{within-sample correlation} of the taxon of interest
\end{tabular}
\caption{The notation for the observed random variables, latent random variables, and parameters of our proposed beta-binomial model. The subscript $i$ refers to the $i^{\text{th}}$ sample.}
\label{tab:notation}
\end{table}

Suppose we have $n$ samples of microbial communities, indexed by $i = 1,\ldots,n$.
Let $M_i$ be the \textit{sequencing depth}, or the number of total counts (or \textit{reads}) across all taxa, in the $i^{\text{th}}$ sample.
Let $Y_{i,j}$ for $j=1,\ldots,M_i$ be an indicator that the $j^{\text{th}}$ read corresponds to the  taxon of interest.
Therefore, $W_i = \sum_{j=1}^{M_i}Y_{i,j}$ is the \textit{observed absolute abundance} of the taxon of interest in the  $i^{\text{th}}$ sample.

It is natural to consider the model
\begin{align}
W_i|(Z_i,M_i) &\sim \text{Binomial}(M_i,Z_i),\label{eq:bin}
\end{align}
and to perform inference on $Z_i$, where $Z_i$ is the probability of observing the taxon of interest in the $i^{\text{th}}$ sample.
However, this model is insufficiently flexible to model microbial abundance data.
For example, Figure~\ref{fig:binomvsbb} (left) shows 95\% prediction intervals from a binomial model fit to the relative abundance of a strain of \textit{Rhizobium} in  16 experimental replicates of sampling microbes in soil (see Section~\ref{s:data} for details). 
We see that the data are substantially overdispersed relative to the binomial model, which provides a very poor fit (see \cite{mcmurdie2014waste} for further discussion on overdispersion of microbial abundance data).

\begin{figure}[ht!]
    \centering
        \includegraphics[width = .95\textwidth]{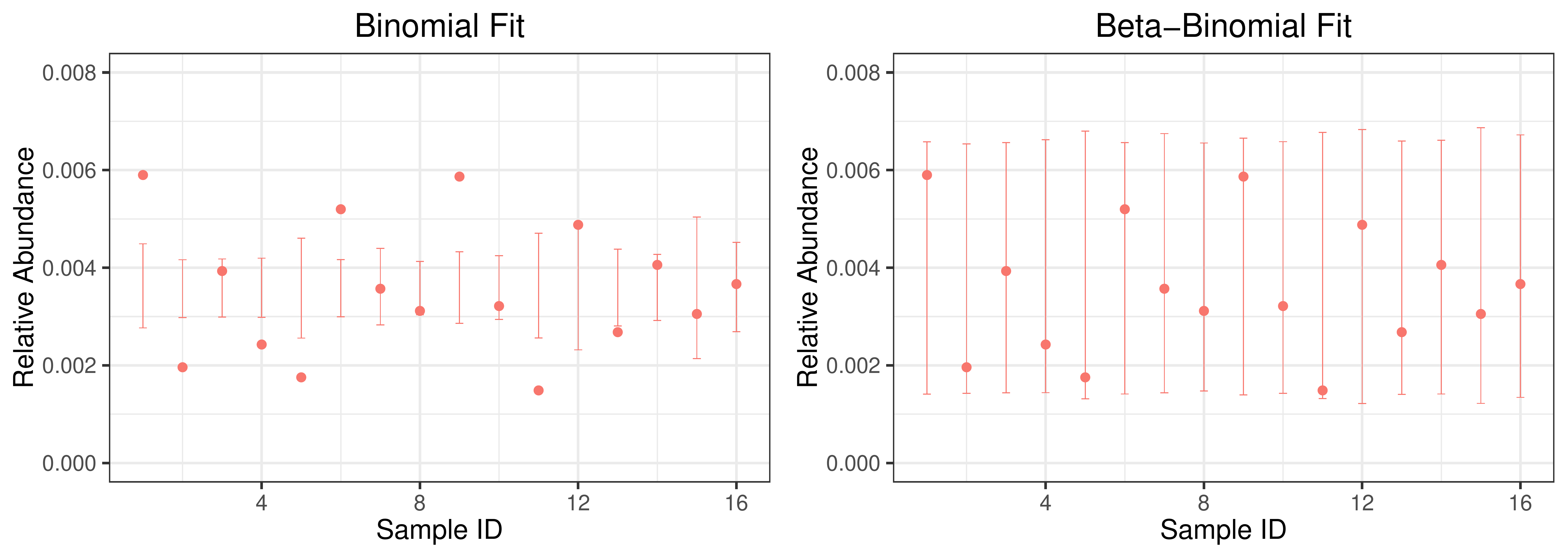}
    \caption{The relative abundance of a strain of \textit{Rhizobium} in 16 biological replicate samples in a soil microbiology study, and 95\% prediction intervals based on a binomial model (left) and  the proposed beta-binomial model (right). The data is clearly overdispersed relative to the binomial model, motivating the development of our beta-binomial model.}
    \label{fig:binomvsbb}
\end{figure}

The overdispersion of the observed relative abundances compared to a binomial model motivates a more flexible model. We propose the following model:
\begin{align}
W_i|(Z_i,M_i) &\sim \text{Binomial}(M_i,Z_i),\label{eq:bb1}\\
Z_i &\sim \text{Beta}(a_{1,i},a_{2,i}), \label{eq:bb2}
 \end{align} 
where $a_{1,i} \in \R_{+}$, $a_{2,i}\in \R_{+}$. 
In the model \eqref{eq:bb1}--\eqref{eq:bb2}, $Z_i$ is itself a random variable, representing the \emph{latent relative abundance} of the taxon.
As we will demonstrate, this hierarchical approach to to modeling relative abundance is a major advantage of our approach. 

Using the parameterization \begin{align} 
\mu_i &=\frac{a_{1,i}}{a_{1,i}+a_{2,i}},\label{eq:mu}
\end{align}
it can be shown that
\begin{align}
\mathbb{E}(W_i|M_i) &= M_i\times \mathbb{E}(Z_i) = M_i\times \mu_i.
\end{align} 
Thus $\mu_i\in (0,1)$ is the \textit{expected relative abundance} of the taxon in the $i^{\text{th}}$ sample. 
In addition, using the parameterization
\begin{align}
\phi_i&=\frac{1}{a_{1,i}+a_{2,i}+1},\label{eq:phi}
\end{align}
it can be shown that
\begin{align}
\text{Var}(W_i|M_i) &= M_i\times \mu_i\times(1-\mu_i)\times(1+(M_i-1)\times\phi_i).
\end{align} 
The multiplicative factor $(1+(M_i-1)\times\phi_i)$ 
is therefore the \textit{overdispersion} of the absolute abundance of the taxon for the $i^{\text{th}}$ sample relative to a binomial random variable.
Furthermore, 
\begin{align}
\text{Corr}(Y_{i,j}, Y_{i,j^*}) &= \phi_i \text{ for } 1\leq j < j^* \leq M_i,
\end{align}
so $\phi_i$ can also be interpreted as the correlation between the taxon indicator variables within the $i^{\text{th}}$ sample \citep{prentice1986binary}.

We then link the expected relative abundance, $\mu_i$, and the overdispersion, $\phi_i$, to covariates.
We define link functions
\begin{align} 
g(\mu_i)&=\beta_0+\bm{X}_{i}^T\boldsymbol{\beta},\label{eq:linkm}\\
h(\phi_i)&=\beta_0^*+\bm{X}^{*T}_{i}\boldsymbol{\beta}^*,\label{eq:linkp}
\end{align}
where $\bm{X}_i$, the $i^{\text{th}}$ row of the covariate matrix $\bm{X}=[X_{ij}]\in \R^{n \times k}$, represents covariates associated with $\mu_i$; $\bm{X}_i^*$, the $i^{\text{th}}$ row of the covariate matrix $\bm{X}^*=[X^*_{ij}]\in \R^{n \times k^*}$, represents covariates associated with $\phi_i$; $\boldsymbol{\beta}=(\beta_1,\ldots,\beta_k)^T$; and $\boldsymbol{\beta}^* = (\beta^*_1,\ldots,\beta_{k^*}^*)^T$.
$\bm{X}$ and $\bm{X}^*$ may be identical, or they may be non- or partially-overlapping.

Throughout this paper, we choose the logit transformation for the link functions in \eqref{eq:linkm} and \eqref{eq:linkp}, so that 
\begin{align*}
    g(x)\equiv h(x) &\defeq \log \paren{\dfrac{x}{1-x}}.
\end{align*}
This link function is convenient as it is a bijection between $[0,1]$ and $\R$.  Other choices for the link functions can be used as well, and the link functions for $\mu_i$ and $\phi_i$ need not be identical.
 
 This hierarchical model has three key advantages over other approaches.
 Firstly, the use of a beta random variable as a model for the binomial probability allows us to incorporate overdispersion.
 Secondly, the overdispersion parameter (rather than just the mean) can be modeled with covariates.
 As we will see in Section~\ref{s:data}, this is a key advantage of our approach.
 Finally, our model makes direct use of the absolute abundance ($W_1,\ldots,W_n$) and the total number of counts ($M_1,\ldots,M_n$), rather than simply transforming these quantities into the observed relative abundance ($W_1/M_1,\ldots,W_n/M_n$), which would amount to throwing away valuable information about the sequencing depth across in each sample. 
 We show the 95\% prediction intervals from a beta-binomial model for the soil microbiology study in Figure~\ref{fig:binomvsbb} (right).

\subsection{Model Fitting}\label{s:est}
Given $n$ samples from the model  \eqref{eq:bb1}--\eqref{eq:bb2}, the log-likelihood is
\begin{align}
    &\log L(\boldsymbol{\theta}|\bm{W},\bm{M})\label{eq:logl}\\
    &\hspace{.5em} = \sumi \log \brack{\bmat M_i\\  W_i\emat \frac{B(a_{1,i}+W_i,\ a_{2,i}+M_i-W_i)}{B(a_{1,i},\ a_{2,i})}}\notag \\
    &\hspace{.5em} =\sumi\log \brack{\bmat M_i\\  W_i\emat \frac{B\paren{ \dfrac{e^{-\beta_0^*-\bm{X}_i^{*T}\boldsymbol{\beta}^*}}{1+e^{-\beta_0-\bm{X}_i^T\boldsymbol{\beta}}}+W_i,\  \dfrac{e^{-\beta_0^*-\bm{X}_i^{*T}\boldsymbol{\beta}^*}}{1+e^{\beta_0+\bm{X}_i^T\boldsymbol{\beta}}}+M_i-W_i}}{B\paren{ \dfrac{e^{-\beta_0^*-\bm{X}_i^{*T}\boldsymbol{\beta}^*}}{1+e^{-\beta_0-\bm{X}_i^T\boldsymbol{\beta}}},\  \dfrac{e^{-\beta_0^*-\bm{X}_i^{*T}\boldsymbol{\beta}^*}}{1+e^{\beta_0+\bm{X}_i^T\boldsymbol{\beta}}}}}}\notag, 
\end{align}
where $\bm{W}\in \R^n$, $\bm{M}\in \R^n$, $\boldsymbol{\beta}\in \R^{k}$, $\boldsymbol{\beta}^* \in \R^{k^*}$, $\boldsymbol{\theta}=(\beta_0, \boldsymbol{\beta}^T, \beta_0^*, \boldsymbol{\beta}^{*T})^T$, and $B(\cdot,\cdot)$ is the Beta function given by $B(x,y) = \int_0^{1} t^{x-1}(1-t)^{y-1}dt$ for $x\in \R$ and $y \in \R_{+}$. 
We fit the model by maximum likelihood using the trust region optimization algorithm \citep{fletcher1987practical, nocedal1999springer, trust}, which has accelerated computation relative to a line search method.

In this iterative algorithm, a ``trust region'' is defined around the parameter estimate at each iteration.
The algorithm then updates the parameter estimate by minimizing a second-order Taylor series expansion of the objective function, subject to the constraint that the solution is within the trust region.
If a proposed update is infeasible (i.e. it is outside of the parameter space), then it is rejected and the trust region shrinks.
The minimization of the objective function then repeats with the new constraint.
If a proposed update is close to the boundary of the trust region, the trust region expands in the next iteration.
We implement the trust algorithm for minimizing the negative log-likelihood using the \texttt{R} package \texttt{trust} \citep{trust}.

The log-likelihood is not concave in $\boldsymbol{\theta}$ (see Appendix~\ref{A:concave}), so trust region optimization does not guarantee convergence to the global minimum of the objective function.
However, under mild conditions, the limit points of the trust algorithm are guaranteed to satisfy the first- and second-order conditions that are necessary for a local minimum \citep{fletcher1987practical, nocedal1999springer}.
We use multiple initializations and select the estimate that has the largest log-likelihood.
In practice, there is little difference in the parameter estimates across initializations.

Each iteration of the trust region optimization algorithm makes use of the gradient and Hessian of \eqref{eq:logl}. 
These are given in Appendix~\ref{A:hessian} for the case of logit link functions for $g(\cdot)$ and $h(\cdot)$ in \eqref{eq:linkm} and \eqref{eq:linkp}.

\section{Hypothesis Testing}\label{s:hyp}

We now discuss inference on $\boldsymbol{\theta}$.
We consider the null hypothesis that $\bm{A}\boldsymbol{\theta}=\bm{b}$, where $\bm{A}\in \R^{r \times (k+k^*+2)}$ has full row rank and $r < k+k^*+2$, $\bm{b} \in \R^{r}$, and where $\boldsymbol{\theta}$ is the parameter vector introduced in \eqref{eq:logl}. 
The Wald test statistic is
\begin{equation}\label{eq:wts}
    \hat{T}_{Wald} = n(\bm{A}\hat{\boldsymbol{\theta}}-\bm{b})^T(\bm{A} \hat{\mathcal{I}}(\hat{\boldsymbol{\theta}})_n^{-1}\bm{A}^T)^{-1}(\bm{A}\hat{\boldsymbol{\theta}}-\bm{b}),
\end{equation}
where 
\begin{equation}\label{eq:theta}
    \hat{\boldsymbol{\theta}} = \argsup_{\boldsymbol{\theta}}\log L(\boldsymbol{\theta}|\bm{W},\bm{M})
\end{equation}  
and  $\hat{\mathcal{I}}(\hat{\boldsymbol{\theta}})_n$ is
the observed Fisher information evaluated at $\hat{\boldsymbol{\theta}}$:
\begin{equation}
   \hat{\mathcal{I}}(\hat{\boldsymbol{\theta}})_n = -\dfrac{1}{n}\sum_{i=1}^{n} \brack{\dfrac{\partial^2}{\partial \boldsymbol{\theta}\partial \boldsymbol{\theta}^T}
    \log L(\boldsymbol{\theta}|\bm{W},\bm{M}) }_{\boldsymbol{\theta}=\hat{\boldsymbol{\theta}}}.
\end{equation}
Under the null hypothesis that $\bm{A}\boldsymbol{\theta}=\bm{b}$, we find empirically that $\hat{T}_{Wald}$ is well-approximated by a $\chi^2_r$ distribution if $n$ is large (Section~\ref{ss:type1}).
Alternatively, we can test $ \bm{A}\boldsymbol{\theta}=\bm{b}$ using a likelihood ratio test statistic, defined as 
\begin{equation}\label{eq:lrts}
 \hat{T}_{LRT} = 2\paren{\log L(\hat{\boldsymbol{\theta}}|\bm{W},\bm{M}) -\log L(\hat{\boldsymbol{\theta}}_0|\bm{W},\bm{M})},
\end{equation}
where 
\begin{equation}\label{eq:nulltheta}
    \hat{\boldsymbol{\theta}}_0 = \argsup_{\boldsymbol{\theta}: \bm{A}\boldsymbol{\theta}=\bm{b}}\log L(\boldsymbol{\theta}|\bm{W},\bm{M}).
\end{equation}
When $n$ is large and $\bm{A}\boldsymbol{\theta}=\bm{b}$, we find that the distribution of 
$\hat{T}_{LRT}$ is well-approximated by a $\chi^2_r$ distribution (Section~\ref{ss:type1}). 

In practice, we often do not have the sample size necessary to use the $\chi^2_r$ approximation.
For this reason, we also implement a parametric bootstrap hypothesis testing procedure.
Our parametric bootstrap Wald testing procedure is given in Algorithm~\ref{alg:pboot}; the parametric bootstrap likelihood ratio test procedure is provided in Appendix \ref{A:pbLRT}.

\begin{algorithm}[ht!]
{\bf Require:} $\bm{W}$, $\bm{M}$, $\bm{X}$, $\bm{X}^*$, a large integer $B$ (e.g. $B=10,000$)
\begin{algorithmic}[1]
\State{Estimate  $\hat{\boldsymbol{\theta}}$ and $\hat{\boldsymbol{\theta}}_0$  as in \eqref{eq:theta} and \eqref{eq:nulltheta}, respectively, with the trust region optimization procedure.}
\State{Compute $\hat{T}_{Wald}$ as in \eqref{eq:wts} using $\bm{A}$, $\bm{b}$, and  $\hat{\boldsymbol{\theta}}$.}
\For{$b=1,\ldots,B$}
\State{Simulate $\tilde{\bm{W}}^b$ with elements $\tilde{W}_i^b$ drawn from a beta-binomial distribution with $M_i$

\noindent\hspace{\algorithmicindent}draws and parameters $\hat{\boldsymbol{\theta}}_0$.}
\State{Estimate $\tilde{\boldsymbol{\theta}}^b$ as in \eqref{eq:theta} using $\tilde{\bm{W}}^b$ and $\bm{M}$ with the trust region optimization procedure.}
\State{Compute $\hat{T}^{b}_{Wald}$ as in \eqref{eq:wts} using $\bm{A}$, $\bm{b}$, and  $\tilde{\boldsymbol{\theta}}^b$.}
 \EndFor
 \State{Calculate the p-value: $$ \hat{p} \leftarrow \dfrac{1}{B+1}\paren{1+\sum_{b=1}^{B}\mathbbm{1}\set{\hat{T}^{b}_{Wald}\geq \hat{T}_{Wald}}}.$$}
 \State{\textbf{return} $\hat{p}$}
\end{algorithmic}
\caption{Parametric Bootstrap Wald Test of $H_0: \ \bm{A}\boldsymbol{\theta}=\bm{b}$}\label{alg:pboot}
\end{algorithm}

For certain realizations of $\bm{W}$, Wald-type inference is uninformative.
\sloppy For example, if $k=k^*=1$, $\bm{X}_{i}=\bm{X}_{i}^* \in \set{0, \ 1}$ for $i=1,\ldots,n$, and $\sum_{i:\ \bm{X}_{i} = 1} W_i = 0$,
 then a parameter estimate diverges to $- \infty$ (see Lemma~\ref{lemma:sub21} in Appendix \ref{A:singularLRT} for details).
This limitation is not unique to our model, and hypothesis testing using Wald tests in the case of complete or quasi-complete separation in logistic regression is known to have the same issue (see \cite{albert1984existence, heinze2002solution, heinze2006comparative} for further discussion). 
In this case, we instead use the likelihood ratio test to test hypotheses about $\boldsymbol{\beta}$, such as $\boldsymbol{\beta} = 0$.
However, in this setting, even the likelihood ratio test does not provide a useful test of certain hypotheses about $\boldsymbol{\beta}^*$, such as $\boldsymbol{\beta}^* = 0$   (see Appendix \ref{A:singularLRT}). 
Since it is often the case that a taxon is unobserved in certain experimental conditions, the default behaviour for our software in this setting is to return a test statistic of zero for Wald-type tests to indicate that inference is uninformative and the null hypothesis should not be rejected.

While \eqref{eq:lrts} and Algorithms~\ref{alg:pboot}--\ref{alg:pbootLRT} hold for any $\bm{A}$ and $\bm{b}$, they require solving \eqref{eq:nulltheta}.
This may be difficult to do for certain $\bm{A}$ and $\bm{b}$. 
In this case, an approximate solution could be obtained by maximizing the likelihood subject to a penalty on $\norm{\bm{A}\boldsymbol{\theta}-\bm{b}}$ (e.g. see \cite{fiacco1968nonlinear, ryan1974penalty}).
Alternatively, approximating the distribution of \eqref{eq:wts} with a $\chi^2_r$ distribution does not require restricted maximum likelihood estimation.

In summary, we implement four hypothesis testing procedures: the Wald test, the likelihood ratio test, the parametric bootstrap Wald test, and the parametric bootstrap likelihood ratio test.
The Wald and likelihood ratio tests permit faster inference than the parametric bootstrap tests.
However, the parametric bootstrap procedures successfully control Type 1 error in small sample sizes.
We now demonstrate the performance of all of these hypothesis testing procedures in simulation.

\section{Simulation Study}\label{s:simstudy}
We now investigate the performance of our approach, which we call \underline{co}unt \underline{r}egressio\underline{n} for \underline{c}orrelated \underline{o}bservations with the \underline{b}eta-binomial, or \texttt{corncob}, under simulation.
We study the Type I error rate and the power when testing for both differential abundance and differential variability.
We generate sequencing depths $\bm{M}\in\R^{n}$ with elements $M_i$ simulated from the empirical distribution of the observed sequencing depths in the data set discussed in Section~\ref{s:data}, which ranges from $7,821$ to $58,655$.
We use sample sizes $n\in \set{10, 30, 100}$ and a binary covariate $\bm{X}_{i}=\bm{X}_{i}^*=0$ for $i=1,\ldots,n/2-1$ and $\bm{X}_{i}=\bm{X}_{i}^*=1$ for $i=n/2,\ldots,n$.
We then simulate absolute abundances $\bm{W}\in \R^{n}$ with elements $W_i$ simulated under the data generating model (described below).
The parameter values were selected by fitting \texttt{corncob} to the genus \textit{Thermomonas} in the data set discussed in Section~\ref{s:data} so that simulated data are similar to what might be observed in a real-world experiment.
For each simulation, we calculate $10,000$ p-values using all four of the hypothesis testing procedures outlined in Section~\ref{s:hyp}: the Wald test, the likelihood ratio test, the parametric bootstrap Wald test, and the parametric bootstrap likelihood ratio test.
We use $1,000$ bootstrap iterations for the parametric bootstrap testing procedures.

\subsection{Type I Error Rate}\label{ss:type1}
We first confirm that \texttt{corncob} controls Type I error at the nominal level.
We generate data using the beta-binomial model with logit link functions for mean and overdispersion, under three settings for $\boldsymbol{\beta}$.
In the first simulation setting, we test the null hypothesis $H_0: (\beta_1, \beta_1^*) = (0,0)$.
We generated model parameters by fitting a model to the genus \textit{Thermomonas} without using soil amendment as a covariate, yielding parameters $(\tilde{\beta}_0, \tilde{\beta}_1, \tilde{\beta}_0^*, \tilde{\beta}_1^*) = (-5.75, 0, -5.24, 0)$.
In the second simulation setting, we test the null hypothesis $H_0: \betost = 0$.
We generated model parameters by fitting a model to the genus \textit{Thermomonas} using soil amendment as a covariate for $\mu_i$, yielding parameters $(\tilde{\beta}_0, \tilde{\beta}_1, \tilde{\beta}_0^*, \tilde{\beta}_1^*) = (-5.36, -1.12, -5.69, 0) $.
In the third simulation setting, we test the null hypothesis $H_0: \beto = 0$.
We generated model parameters by fitting a model to the genus \textit{Thermomonas} using soil amendment as a covariate for $\phi_i$, yielding parameters $(\tilde{\beta}_0, \tilde{\beta}_1, \tilde{\beta}_0^*, \tilde{\beta}_1^*) = (-5.51,0, -5.38,0.70)$.
For all three simulation settings, the null hypotheses are true, so we would expect p-values obtained from testing the null hypotheses to be uniformly distributed.

 The results are shown in Figure~\ref{fig:type1}. 
 For sample sizes of $30$ and $100$, all testing procedures resulted in approximately uniform p-values, and Type I error is controlled.
 This suggests that for this experiment, a sample size of $30$ is sufficient to approximate the distribution of the Wald and likelihood ratio test statistics using a $\chi^2$ distribution.
 
 For a sample size of $10$, only the parametric bootstrap procedures resulted in approximately uniform p-values and successful Type I error control.
 The p-values obtained using the Wald and likelihood ratio tests were anti-conservative, suggesting that for this experiment, a sample size of $10$ is too small to approximate the distribution of the test statistics using a $\chi^2$ distribution.
 Therefore, to obtain reliable inference, we recommend the parametric bootstrap procedure when $n$ is smaller than $30$.

\begin{figure}[ht!]
        \centering
        \includegraphics[width = .95\textwidth]{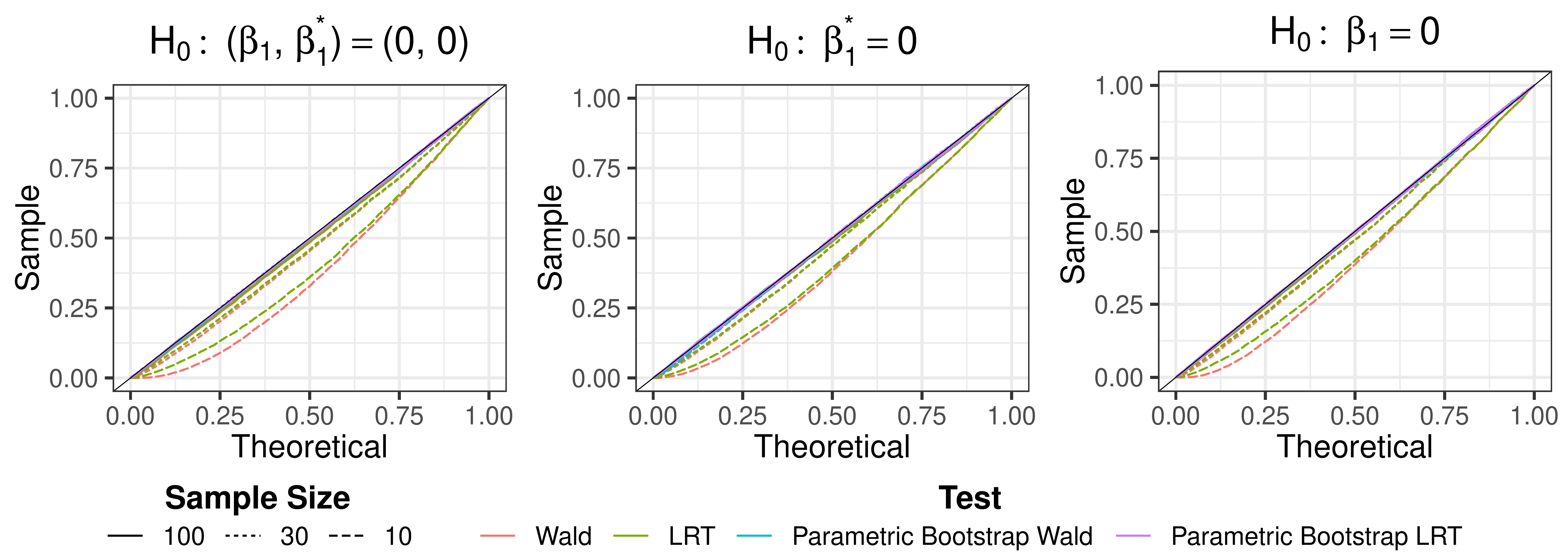}
    \caption{Quantiles of p-values obtained from the Type I error rate simulation settings compared to quantiles of a Uniform$(0,1)$ distribution. We test the null hypotheses $H_0: (\beta_1, \beta_1^*) = (0,0)$ (left), $H_0: \beta_1^* = 0$ (middle), and $H_0: \beta_1 = 0$ (right). A 45-degree line is shown (black). P-values were obtained using Wald (red), likelihood ratio (green), parametric bootstrap Wald (blue), and parametric bootstrap likelihood ratio (purple) tests.
    Sample sizes used were 10 (dashed), 30 (dotted), and 100 (solid).}
    \label{fig:type1}
\end{figure}

\subsection{Power}
We now investigate the power of \texttt{corncob} to reject (i) the null hypothesis $H_0: \ \beta_1 = 0$, as well as (ii) the null hypothesis $H_0: \ \beta_1^* = 0$.
We consider two cases: varying the value of $\beta_1$, and varying the value of $\beta_1^*$. 
For both settings, we generated model parameters by fitting a model to the genus \textit{Thermomonas} using soil amendment as a covariate for $\mu_i$ and $\phi_i$, yielding parameters $(\tilde{\beta}_0, \tilde{\beta}_1, \tilde{\beta}_0^*, \tilde{\beta}_1^*) = (-5.17,        -2.46,        -5.13,        -3.88)$. 
In the first case (Setting 4 in Figure~\ref{fig:power}), we set $(\beta_0, \beta_1, \beta_0^*, \beta_1^*) = (\tilde{\beta}_0, c\tilde{\beta}_1, \tilde{\beta}_0^*, \tilde{\beta}_1^*)$ using  $c \in \set{0, 0.05, \ldots, 1}$.
In the second case (Setting 5 in Figure~\ref{fig:power}), we set $(\beta_0, \beta_1, \beta_0^*, \beta_1^*) = (\tilde{\beta}_0, \tilde{\beta}_1, \tilde{\beta}_0^*, c\tilde{\beta}_1^*)$ using $c \in \set{0, 0.05, \ldots, 1}$.

The results of the power analyses are shown in Figure~\ref{fig:power}.
For both null hypotheses, all sample sizes, and all hypothesis testing procedures, the power increases as both the sample size and the magnitude of the coefficient being tested increases.
For sample sizes of $30$ and $100$, there is little difference in power across the four testing procedures.
This is not surprising, given that in the simulations in Section~\ref{ss:type1}, all procedures performed similarly with sample sizes of $30$ and $100$.
We do not show results for the procedures that rely on the asymptotic distribution of the test statistics for $n=10$, as we saw in Section~\ref{ss:type1} that these procedures did not properly control Type I error.

\begin{figure}[ht!]
        \centering
        \includegraphics[width = .95\textwidth]{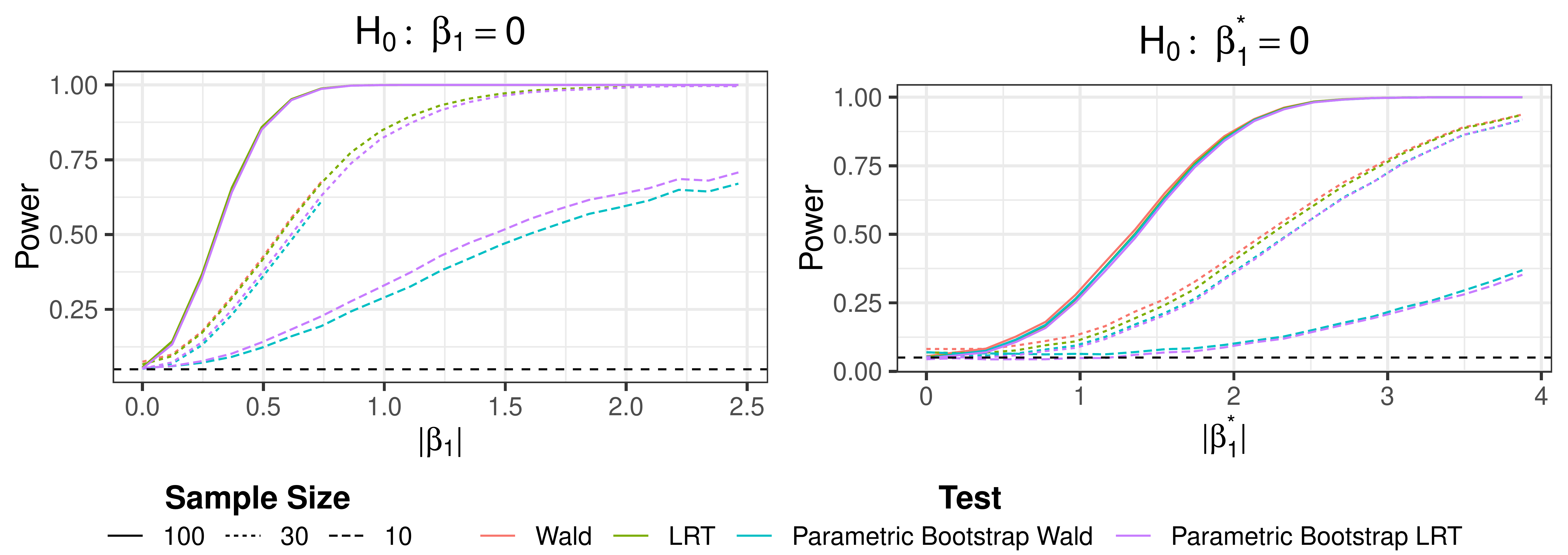}
    \caption{Power curves of p-values obtained from the power simulations. Setting 4 (left) tests $H_0: \beta_1 =0$. Setting 5 (right) tests $H_0: \beta_1^* = 0$.
    A horizontal dashed line is shown at $0.05$. P-values were obtained using Wald (red), likelihood ratio (green), parametric bootstrap Wald (blue), and parametric bootstrap likelihood ratio (purple) tests.
    Sample sizes used were 10 (dashed), 30 (dotted), and 100 (solid).}
    \label{fig:power}
\end{figure}

\section{Application to Soil Data}\label{s:data}

We now consider a study of the association between soil treatments and soil microbiome composition \citep{whitman2016dynamics}.
In this experiment, there are three groups of soil treatments: no additions, biochar additions, and fresh biomass additions.
For each treatment group, multiple experimental replicates were taken at three time points: on the first day, after 12 days, and after 82 days.
The data include  $n=119$ samples with sequencing depths ranging from $8,830$ to $194,356$.
After quality control (as described in \cite{whitman2016dynamics}), a total of $7,770$ operational taxonomic units were identified using the UPARSE workflow \citep{edgar2013uparse}, and taxonomy was assigned using reference databases.
Using the assigned taxonomy, we aggregated counts to the genus level, giving $241$ genera.

We are interested in applying our method to compare the microbiome of soil with no additions after 82 days $(n=15)$ to the microbiome of soil with biochar additions after 82 days $(n=16)$.
We remove $13$ genera for which the total number of counts in these $31$ samples is zero.
We apply \texttt{corncob} using soil addition as a covariate for $\mu_i$ and $\phi_i$ as in \eqref{eq:linkm} and \eqref{eq:linkp}.
We calculate p-values using the parametric bootstrap likelihood ratio test (Algorithm~\ref{alg:pbootLRT}) with $B=10^6$ bootstrap iterations.
We compare the results of \texttt{corncob} to those from DESeq2 \citep{love2014moderated}, EdgeR \citep{robinson2010edger}, metagenomeSeq \citep{paulson2013differential}, and a zero-inflated beta (ZIB) regression model \citep{peng2016zero}.

\subsection{Detection of Differential Abundance}\label{ss:DA}

We first compare p-values obtained from testing for differential abundance across soil addition group.
Roughly speaking, each of the approaches tests for a difference in abundance of a single taxon across conditions, although the details of the model used vary across methods.
In the context of \texttt{corncob}, testing for differential abundance amounts to testing the null hypothesis $H_0: \ \boldsymbol{\beta} = 0$, using the notation defined in \eqref{eq:linkm}.
Scatter plots of the negative log-10 p-values for each approach are shown in Figure~\ref{fig:DA}.

\begin{figure}[ht!]
    \centering
        \includegraphics[width = .98\textwidth]{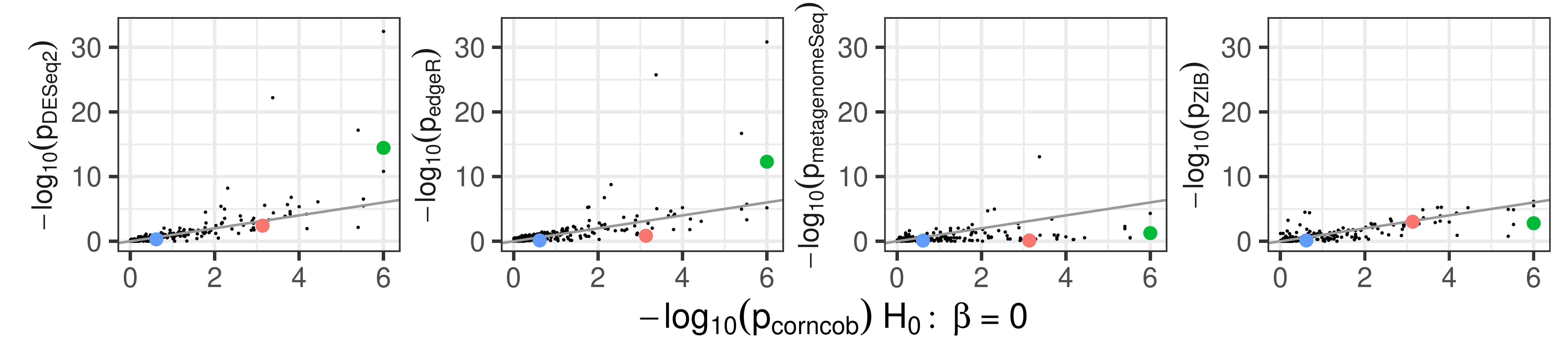}
    \caption{The negative log-10 p-values obtained by testing for differential abundance using \texttt{corncob} ($H_0: \boldsymbol{\beta} = 0$) compared to those from DESeq2 (left-most, Spearman's correlation coefficient $\rho  = 0.854$), EdgeR (middle-left, $\rho = 0.783$), metagenomeSeq (middle-right, $\rho =   0.552$), and ZIB (right-most, $\rho = 0.705$). A 45-degree line is shown. We see that the p-values are on a similar scale overall. \textit{Thermomonas} (green), \textit{Flavisolibacter} (red), and \textit{Myxococcus} (blue) are further examined in Figure~\ref{fig:dvex}. }
    \label{fig:DA}
\end{figure}

Overall, as p-values calculated using \texttt{corncob} decrease, so do those calculated using other approaches.
We observe moderate to strong correlations across the different approaches, with Spearman's correlation coefficients 
between the p-values obtained from \texttt{corncob} ($H_0: \ \boldsymbol{\beta}=0$) and DESeq2, edgeR, metagenomeSeq, and ZIB, respectively, of 0.854, 0.783, 0.552, 0.705.
\texttt{corncob} calculated a lower p-value for 53.9\%, 43.6\%, 63.8\%, and 58.3\% of genera compared to DESeq2, edgeR, metagenomeSeq, and ZIB, respectively. 
Median p-values across all genera for \texttt{corncob}, DESeq2, edgeR, metagenomeSeq, and ZIB are 0.273, 0.318, 0.297, 0.491, and 0.320, respectively.
Therefore, while the p-values produced by \texttt{corncob} are on a similar scale to the other approaches, they may be higher or lower for any given taxon.
 While each of the approaches uses a different model and makes use of a different test statistic, they are all testing for some difference in the mean abundance of the taxon across the soil addition.
 Thus, it is unsurprising that the p-values are similar across the approaches.

\subsection{Detection of Differential Variability}\label{ss:DV}

We now test for differences in the variability of the abundance of a single taxon across conditions, which we refer to as \textit{differential variability}. 
Using \texttt{corncob} and the notation in \eqref{eq:linkm}--\eqref{eq:linkp}, this amounts to testing the null hypothesis $H_0:\ \boldsymbol{\beta}^* =0$.
As far as we know, \texttt{corncob} is the only approach that explicitly tests for differential variability.
Thus, in this section, we investigate whether testing for differential variability allows us to identify new genera beyond what we identify when testing only for differential abundance. 

We compare the results of testing for differential variability to the results of testing for differential abundance using the methods investigated in Section~\ref{ss:DA}.
Figure~\ref{fig:DV} shows scatter plots of the negative log-10 transformations of the p-values for testing differential abundance from DESeq2, metagenomeSeq, ZIB, and \texttt{corncob} against the p-values for testing differential variability with \texttt{corncob}.
We see from Figure~\ref{fig:DV} that there is only a weak association between the p-values for differential variability obtained using \texttt{corncob} and the p-values for differential abundance obtained using the other approaches.
In particular, Spearman's correlation coefficients are 0.127, 0.234, 0.132, 0.215, and 0.362 between \texttt{corncob} p-values for $H_0: \ \boldsymbol{\beta}^*=0$ and p-values from DESeq2, edgeR, metagenomeSeq, ZIB, and \texttt{corncob} for $H_0: \ \boldsymbol{\beta}=0$, respectively.
We omit from Figure~\ref{fig:DV} the scatter plot comparing the \texttt{corncob} p-values for $H_0:\ \boldsymbol{\beta}^*=0$ to the edgeR p-values because the p-values from edgeR are similar to those from DESeq2.
We conclude that applying \texttt{corncob} to test $H_0: \ \boldsymbol{\beta}^*=0$ leads to the discovery of a very different set of genera than those discovered by applying \texttt{corncob} or other approaches to test for differential abundance.

\begin{figure}[ht!]
    \centering
        \includegraphics[width = .98\textwidth]{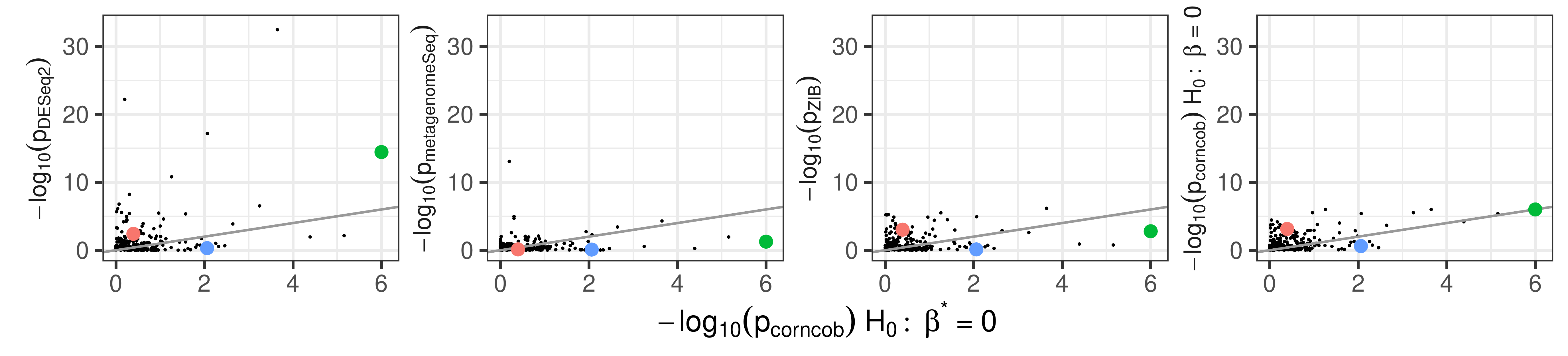}
    \caption{The negative log-10 p-values obtained by testing for differential variability using \texttt{corncob} ($H_0:\ \boldsymbol{\beta}^*=0$) compared to the negative log-10 p-values obtained by testing for differential abundance using DESeq2 (left-most, Spearman's correlation coefficient $\rho = 0.127$),  metagenomeSeq (middle-left, $\rho = 0.132$), ZIB (middle-right, $\rho =  0.215$), and \texttt{corncob} ($H_0:\ \boldsymbol{\beta}=0$) (right-most, $\rho =  0.362$). A 45-degree line is shown.  \textit{Thermomonas} (green), \textit{Flavisolibacter} (red), and \textit{Myxococcus} (blue) are further examined in Figure~\ref{fig:dvex}. We omit a scatter plot showing p-values for edgeR ($\rho = 0.234$); results are similar to DESeq2.}
    \label{fig:DV}
\end{figure}

To obtain greater insight into the results shown in Figure~\ref{fig:DV}, we consider the 3 highlighted genera, which we further investigate in Figure~\ref{fig:dvex}. 
The first, \textit{Thermomonas}, has small p-values for both differential abundance ($p=1.00\times 10^{-6}$) and differential variability ($p=1.00\times 10^{-6}$) using \texttt{corncob}. 
The second, \textit{Flavisolibacter}, has a small p-value for differential abundance ($p=7.44\times 10^{-4}$) and a large p-value for differential variability ($p=0.404$).
The third, \textit{Myxococcus}, has a large p-value for differential abundance ($p=0.244$) and a small p-value for differential variability ($p=8.83\times 10^{-3}$), so it would not be identified using the competing approaches (see Figure~\ref{fig:dvex} for p-values from all approaches).
Figure~\ref{fig:dvex} indicates a clear visual difference between genera that are  identified as differentially abundant but not differentially variable,  differentially variable but not differentially abundant, and both differentially abundant and differentially variable. 
Researchers can use \texttt{corncob} to distinguish between these three possibilities.

\begin{figure}[ht!]
    \centering
        \includegraphics[width = .95\textwidth]{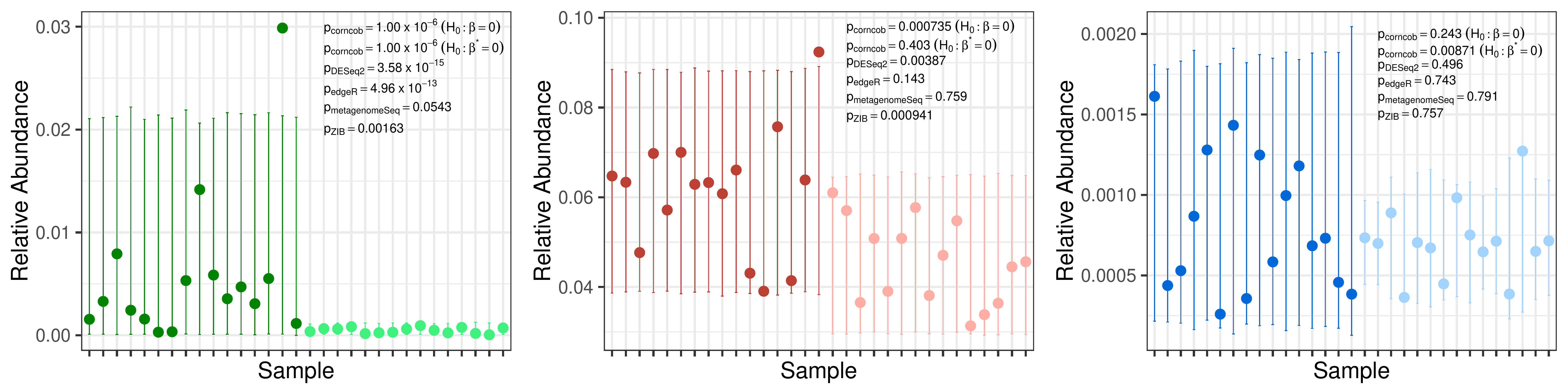}
    \caption{The observed relative abundances of the genera \textit{Thermomonas} (left), \textit{Flavisolibacter} (middle), and \textit{Myxococcus} (right) in 31 soil samples.
    Each of these genera is highlighted in each panel of Figures~\ref{fig:DA} and \ref{fig:DV}. In each panel, the first 16 samples correspond to the biochar additions group (darker color), and the remaining 15 samples correspond to the no additions group (lighter color). 95\% prediction intervals for the relative abundances from a \texttt{corncob} fit using soil addition as a covariate for $\mu_i$ and $\phi_i$ are shown. Using \texttt{corncob} to test $H_0:\  \boldsymbol{\beta} = 0 $ and $H_0:\  \boldsymbol{\beta}^* = 0 $ indicates that \textit{Thermomonas} is both differentially abundant ($p=1.00\times 10^{-6}$) and differentially variable ($p=1.00\times 10^{-6}$), \textit{Flavisolibacter} is differentially abundant ($p=7.44\times 10^{-4}$) and not differentially variable ($p=0.404$), and \textit{Myxococcus} is differentially variable ($p=8.83 \times 10^{-3}$) and not differentially abundant ($p=0.244$).}
    \label{fig:dvex}
\end{figure}

In practice, a data analyst will apply a multiple testing procedure to adjust the p-values for multiple comparisons, so we also investigative the number of genera identified as either differentially abundant or differentially variable after applying the Benjamini-Hochberg procedure \citep{benjamini1995controlling} to the p-values obtained using \texttt{corncob} to test $H_0: \ \boldsymbol{\beta}=0$ and $H_0: \ \boldsymbol{\beta}^*=0$.
The results are shown in Figure~\ref{fig:fdr}.
We see that for a given false discovery rate, in this data set we detect more genera as being differentially abundant than differentially variable; this can also be seen in the right-most panel of Figure~\ref{fig:DV}.
All code for performing this analysis is available in the supplementary materials available at \url{github.com/bryandmartin/corncob_supplementary}.

\begin{figure}[ht!]
        \centering
        \includegraphics[width = .7\textwidth]{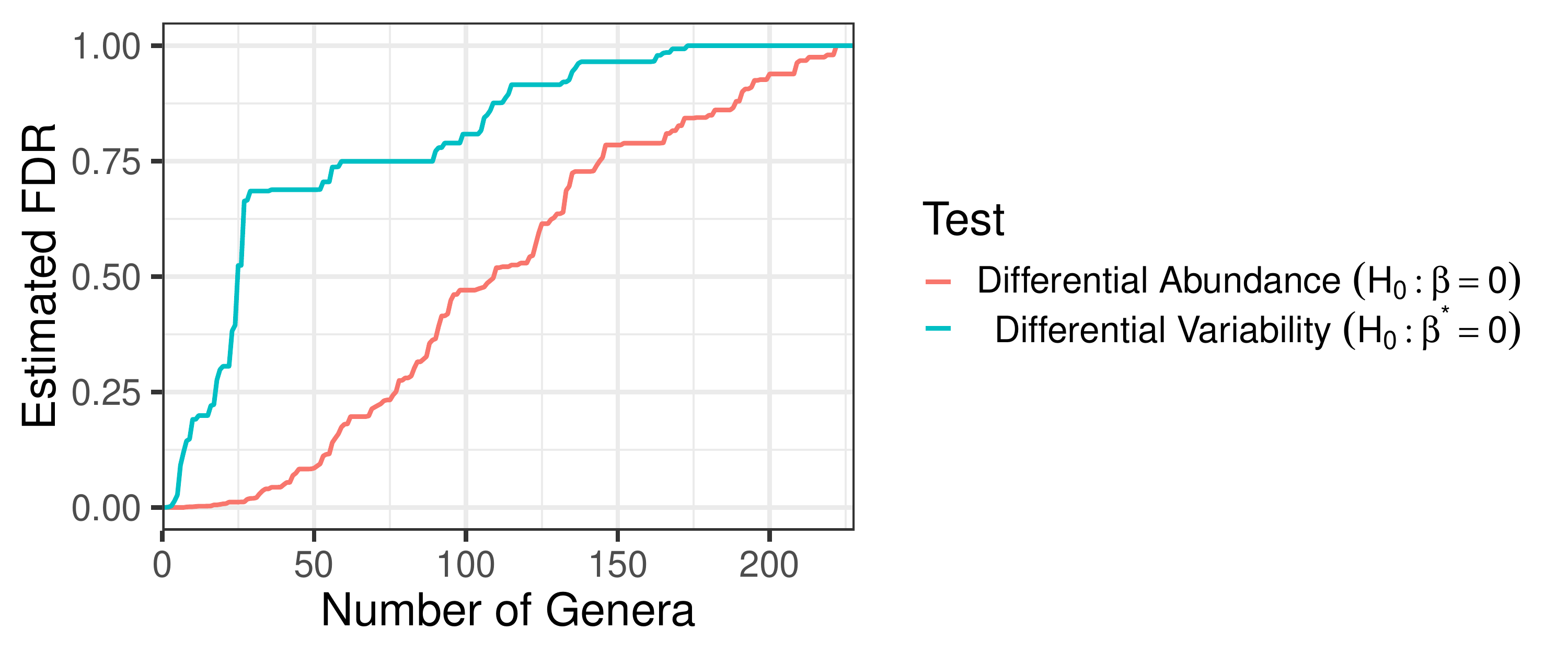}
    \caption{The estimated false discovery rate using the Benjamini-Hochberg procedure, as a function of the number of genera identified as differentially abundant and differentially variable. For a given false discovery rate, we identify fewer genera that are differentially variable than differentially abundant.}
    \label{fig:fdr}
\end{figure}

\section{Discussion}\label{s:disc}

In this paper, we have proposed a beta-binomial regression model for abundance data.
Our model extends existing beta-binomial models by allowing discrete and continuous covariates to be linked to both a relative abundance parameter and an overdispersion parameter. 
Our method is particularly well-suited to modeling microbial abundance data for a number of reasons.
First, microbial taxa are commonly unobserved in many samples.
For example, in the data set examined in Section~\ref{s:data}, 34\% of absolute abundances were zero.
Our model can accommodate this without requiring a zero-inflation component or pseudo-counts.
Second, studies of microbial populations often have small sample sizes.
Our simulation study in Section~\ref{s:simstudy} suggests that our parametric bootstrap inference methods (Algorithms~\ref{alg:pboot} and \ref{alg:pbootLRT}) give valid inference even with small samples. 
Third, the interpretation of $\mu_i$ as the expected relative abundance and of $\phi_i$ as the within-sample correlation of taxon labels (i.e. $\phi_i=\text{Corr}(Y_{ij}=Y_{ij'})$, see Section~\ref{s:meth}) are intuitive and complement ecological theory \citep{welch2016biogeography}.
Finally, regression models for contrasting microbial populations commonly focus on differential abundance. 
By conducting inference about $\phi_i$, our model is also able to identify differences in microbial populations associated with differential variability.

Many studies (e.g. see \cite{gerber2014dynamic, faust2015metagenomics, zhou2015longitudinal}, among others) employ a longitudinal design to investigate the dynamics of microbial populations over time.
To accommodate this setting, future work could incorporate random effects into \eqref{eq:linkm} and \eqref{eq:linkp}.

Our proposed approach models a single taxon's abundance.
A limitation of this approach is that it does not enforce the compositionality constraint (i.e. the estimated expected relative abundances need not sum to $1$ across all microbes in the population).
Future work could consider a multivariate extension of our approach to enforce the compositionality constraint or incorporate between-taxon correlations.

All methods proposed in this paper are implemented in an \texttt{R} package available at \url{github.com/bryandmartin/corncob}. 
Code to reproduce all simulations and data analyses are available at \url{github.com/bryandmartin/corncob_supplementary}.

\bibliography{bib}

\appendix
\section{Analytic Expressions  for the Gradient and Hessian}\label{A:hessian}
Let $\gamma_i = \frac{\phi_i}{1-\phi_i}$ for all $i$, and define $\psi(x)=\int_0^\infty\paren{\frac{e^{-t}}{t}-\frac{e^{-xt}}{1-e^{-t}}}dt$ for $x\in \R_+$ to be the digamma function, the derivative of the logarithm of the gamma function.
Define $\bm{Z}_i = \bmat 1 & \bm{X}_i \emat$ and $\bm{Z}_i^* = \bmat 1 & \bm{X}_i^* \emat$ to be the design matrices for covariates associated with $\mu_i$ and $\phi_i$, respectively, including intercept terms.
Then the expression for the gradient of \eqref{eq:logl} is given by
\begin{align}
    \dfrac{\partial \log L(\boldsymbol{\theta}|\bm{W},\bm{M})}{\partial \boldsymbol{\beta}}&= \sumi \bigg\{\gamma_i^{-1}\mu_i(1-\mu_i)\mathbf{Z}_{i} \bigg[\psi\paren{\dfrac{1-\mu_i}{\gamma_i}}\label{eq:grad1}\\
    &\tab -\psi\paren{M_i+\dfrac{1-\mu_i-W_i\gamma_i}{\gamma_i}} \notag \\
    &\tab+ \psi\paren{W_i + \dfrac{\mu_i}{\gamma_i}}-\psi\paren{\dfrac{\mu_i}{\gamma_i}}\bigg]\bigg\},\notag\\
    \dfrac{\partial \log L(\boldsymbol{\theta}|\bm{W},\bm{M})}{\partial \boldsymbol{\beta}^*}&= \sumi \bigg\{\gamma_i^{-1}\mathbf{Z}^*_{i} \bigg[\psi\paren{M_i+\dfrac{1}{\gamma_i}}- \psi\paren{\dfrac{1}{\gamma_i}}+(\mu_i-1) \label{eq:grad2} \\
    &\tab \bigg(\psi\paren{M_i+\dfrac{1-\mu_i-W_i\gamma_i}{\gamma_i}}-\psi\paren{\dfrac{1-\mu_i}{\gamma_i}}\bigg) \notag \\
    &\tab +\mu_i \paren{\psi\paren{\dfrac{\mu_i}{\gamma_i}}-\psi\paren{W_i+\dfrac{\mu_i}{\gamma_i}}}\bigg]\bigg\} \notag.
\end{align}

Let $\psio(x)=\dfrac{\partial}{\partial x} \psi(x)$ be the trigamma function.
Define $\bm{Y}_i = \bmat \bm{Z}_i^T &  \bm{0}\emat^T\in \R^{k+k^*+2}$ and $\bm{Y}_i^* = \bmat\bm{0} &  \bm{Z}_i^{*T}\emat^T \in \R^{k+k^*+2}$.
Then the expression for the Hessian of \eqref{eq:logl}, $\bm{H}$, is given by
\bal 
\bm{H} = \sum_{i=1}^{n} &\bigg[c_{1,i} \mu_i^2 (1-\mu_i)^2 \bm{Y}_i \bm{Y}_i^T + c_{2,i} \left(\mu_i (1-\mu_i) \bm{Y}_i \gamma_i \bm{Y}_i^{*T}\right.\\
&\tab \left.+ \gamma_i \bm{Y}_i^*  \mu_i(1-\mu_i)\bm{Y}_i^T\right) + c_{3,i} \paren{\gamma_i \bm{Y}_i^* \gamma_i \bm{Y}_i^{*T}}\\
&\tab + c_{4,i} \paren{\mu_i (1-\mu_i)(1-2\mu_i) \bm{Y}_i \bm{Y}_i^T  } +  c_{5,i} \paren{\gamma_i\bm{Y}_i^*  \bm{Y}_i^{*T} }\bigg],
\eal 
where
\bal 
c_{1,i} &= \Big[\psio\paren{M_i+(1-\mu_i-W_i\gamma_i)/\gamma_i}-\psio\paren{(1-\mu_i)/\gamma_i} \\
&\tab + \psio\paren{W_i + \mu_i/\gamma_i}-\psio\paren{\mu_i/\gamma_i}\Big]\gamma_i^{-2},\\
c_{2,i} &= \Big[ \gamma_i(\psi(M_i-(\mu_i+W_i\gamma_i -1)/\gamma_i) - \psi((1-\mu_i)/\gamma_i))+ \gamma_i(\psi(\mu_i/\gamma_i)\\
&\tab - \psi(\mu_i/\gamma_i + W_i))+ (\mu_i-1) (\psio((1-\mu_i)/\gamma_i) \\
&\tab  - \psio(M_i-(\mu_i+W_i\gamma_i-1)/\gamma_i)) +\psio(\mu_i/\gamma_i)\\ 
&\tab  - \psio(\mu_i/\gamma_i+W_i)\Big]\gamma_i^{-3},\\
c_{3,i} &= \Big[ 2\gamma_i\psi(1/\gamma_i)+ \psio(1/\gamma_i) -2\gamma_i \psi(M_i+1/\gamma_i) - \psio(M_i+1/\gamma_i)\\ 
&\tab +(\mu_i-1)^2 \psio(M_i-(\mu_i+W_i\gamma_i -1)/\gamma_i) \\
&\tab - 2\gamma_i(\mu_i-1)\psi(M_i - (\mu_i+W_i\gamma_i -1)/\gamma_i)-\mu_i^2 \psio(\mu_i/\gamma_i)\\
&\tab +\mu_i^2\psio(\mu_i/\gamma_i+W_i) -(\mu_i-1)^2 \psio((1-\mu_i)/\gamma_i)  \\
&\tab + 2\gamma_i(\mu_i-1)\psi((1-\mu_i)/\gamma_i)-2\gamma_i\mu_i \psi(\mu_i/\gamma_i)\\
&\tab + 2\gamma_i\mu_i\psi(\mu_i/\gamma_i +W_i) \Big]\gamma_i^{-4},\\
c_{4,i} &= \Big[\psi((1-\mu_i)/\gamma_i) - \psi(M_i-(\mu_i +W_i\gamma_i -1)/\gamma_i) + \psi(\mu_i/\gamma_i +W_i)\\
&\tab - \psi(\mu_i/\gamma_i) \Big]\gamma_i^{-1},\\
c_{5,i} &= \Big[\psi(M_i+1/\gamma_i) - \psi(1/\gamma_i) + \mu_i (\psi(\mu_i/\gamma_i)-\psi(\mu_i/\gamma_i + W_i))\\
&\tab +(\mu_i-1)(\psi(M_i-(\mu_i+W_i\gamma_i-1)/\gamma_i) - \psi((1-\mu_i)/\gamma_i))\Big]\gamma_i^{-2}.
\eal 

\section{Non-Concavity of the Beta-Binomial Log-likelihood}\label{A:concave}

We show that \eqref{eq:logl} is not guaranteed to be concave in $\boldsymbol{\theta}$. 
Let $n=1$, $\bm{W} \equiv W_1 = 15$, and $\bm{M} \equiv M_1 = 2000$. 
Suppose further that $\boldsymbol{\theta}=(\beta_0, \beta_0^*)^T$. 
Let $\boldsymbol{\theta}^1 = \bmat -3 & -5 \emat^T$ and $\boldsymbol{\theta}^2 = \bmat -1 & -5 \emat^T$.
Then
\bal 
\log L(\boldsymbol{\theta}^1|W_1, M_1) &= -8.481,\\
\log L(\boldsymbol{\theta}^2|W_1, M_1) &= -9.816,\\
\log L(0.5\boldsymbol{\theta}^1 + 0.5\boldsymbol{\theta}^2|W_1, M_1) &= -9.251.
\eal 
Therefore there exists $\boldsymbol{\theta}^1, $ and $\boldsymbol{\theta}^2$ such that
$$\log L(0.5\boldsymbol{\theta}^1 + 0.5\boldsymbol{\theta}^2|\bm{W}, \bm{M})< 0.5 \log L(\boldsymbol{\theta}^1|\bm{W}, \bm{M}) + 0.5 \log L(\boldsymbol{\theta}^2|\bm{W}, \bm{M}),  $$
which establishes that \eqref{eq:logl} is not concave in $\boldsymbol{\theta}$.

\section{Parametric Bootstrap Likelihood Ratio Test}\label{A:pbLRT}

We present Algorithm~\ref{alg:pbootLRT} to conduct a parametric bootstrap likelihood ratio test.

\begin{algorithm}[H]
{\bf Require:} $\bm{W}$, $\bm{M}$, $\mathbf{X}$, $\mathbf{X}^*$, a large integer $B$ (e.g. $B=10,000$)
\begin{algorithmic}[1]
\State{Estimate  $\hat{\boldsymbol{\theta}}$ and $\hat{\boldsymbol{\theta}}_0$  as in \eqref{eq:theta} and \eqref{eq:nulltheta}, respectively, with the trust region optimization procedure.}
\State{Compute $\hat{T}_{LRT}$ as in \eqref{eq:lrts} using $\bm{W}$, $\bm{M}$,  $\hat{\boldsymbol{\theta}}$, and $\hat{\boldsymbol{\theta}}_0$.}
\For{$b=1,\ldots,B$}
\State{Simulate $\tilde{\bm{W}}^b$ with elements $\tilde{W}_i^b$ drawn from a beta-binomial distribution with $M_i$

\noindent\hspace{\algorithmicindent}draws and parameters $\hat{\boldsymbol{\theta}}_0$.}
\State{Estimate $\tilde{\boldsymbol{\theta}}^b$ as in \eqref{eq:theta} using $\tilde{\bm{W}}^b$ and $\bm{M}$ with the trust region optimization procedure.}
\State{Estimate $\tilde{\boldsymbol{\theta}}_0^b$ as in \eqref{eq:nulltheta} using $\tilde{\bm{W}}^b$ and the trust region optimization procedure.}
\State{Compute $\hat{T}^{b}_{LRT}$ as in \eqref{eq:lrts} using $\tilde{\bm{W}}^b$, $\bm{M}$, $\tilde{\boldsymbol{\theta}}^b$, and $\tilde{\boldsymbol{\theta}}_0^b$.}
 \EndFor
 \State{Calculate the p-value: $$ \hat{p} \leftarrow \dfrac{1}{B+1}\paren{1+\sum_{b=1}^{B}\mathbbm{1}\set{\hat{T}^{b}_{LRT}\geq \hat{T}_{LRT}}}.$$}
 \State{\textbf{return} $\hat{p}$}
\end{algorithmic}
\caption{Parametric Bootstrap Likelihood Ratio Test of $H_0: \ \bm{A}\boldsymbol{\theta}=\bm{b}$}\label{alg:pbootLRT}
\end{algorithm}

\section{Likelihood Ratio Testing with a Zero-Count Group}\label{A:singularLRT}

We prove that testing the null hypothesis $H_0: \ \boldsymbol{\beta}^* = 0$ results in a test statistic of zero under certain conditions.
We first prove in Lemma~\ref{lemma:sub21} that the log-likelihood of the model \eqref{eq:bb1}--\eqref{eq:bb2} is equal to zero under certain conditions.
We use this to prove our main claim in Theorem~\ref{theorem:2}.

\begin{lemma}\label{lemma:sub21}
Consider the model \eqref{eq:bb1}--\eqref{eq:bb2} with parameters as in \eqref{eq:mu}--\eqref{eq:phi} and link functions as in  \eqref{eq:linkm}--\eqref{eq:linkp} in the simplified setting with no covariates for $\mu_i$, so that $\boldsymbol{\theta} =(\beta_0, \beta_0^*, \boldsymbol{\beta}^{*T})^T$.
Suppose that $\sum_i W_i = 0$. 
Then
$$\sup_{\beta_0} \log L(\boldsymbol{\theta}|\bm{W},\bm{M})=0.$$
\end{lemma}
\begin{proof}
We write the log-likelihood
\bal 
&\log L(\boldsymbol{\theta}|\bm{W},\bm{M})= \sumi\log \brack{\bmat M_i\\  W_i\emat \frac{B\paren{ \dfrac{e^{-\beta_0^*-\bm{X}_i^{*T}\boldsymbol{\beta}^*}}{1+e^{-\beta_0}}+W_i,\  \dfrac{e^{-\beta_0^*-\bm{X}_i^{*T}\boldsymbol{\beta}^*}}{1+e^{\beta_0}}+M_i-W_i}}{B\paren{ \dfrac{e^{-\beta_0^*-\bm{X}_i^{*T}\boldsymbol{\beta}^*}}{1+e^{-\beta_0}},\  \dfrac{e^{-\beta_0^*-\bm{X}_i^{*T}\boldsymbol{\beta}^*}}{1+e^{\beta_0}}}}}.
\eal 
Substituting $W_i=0$, using the definition of $B(\cdot,\cdot)$, and taking the limit in $\beta_0$ gives
\bal 
\lim_{\beta_0\to-\infty} \log L(\boldsymbol{\theta}|\bm{W},\bm{M}) &=\lim_{\beta_0\to-\infty}\sumi \log \brack{\Gamma\paren{\dfrac{e^{-\beta_0^*-\bm{X}_i^{*T}\boldsymbol{\beta}^*}}{1+e^{\beta_0}}+M_i}} \\
&\tab + \log \brack{\Gamma\paren{\dfrac{e^{-\beta_0^*-\bm{X}_i^{*T}\boldsymbol{\beta}^*}}{1+e^{-\beta_0}}+ \dfrac{e^{-\beta_0^*-\bm{X}_i^{*T}\boldsymbol{\beta}^*}}{1+e^{\beta_0}}}}\\
&\tab - \log \brack{\Gamma\paren{\dfrac{e^{-\beta_0^*-\bm{X}_i^{*T}\boldsymbol{\beta}^*}}{1+e^{-\beta_0}} + \dfrac{e^{-\beta_0^*-\bm{X}_i^{*T}\boldsymbol{\beta}^*}}{1+e^{\beta_0}} + M_i}}\\
&\tab - \log \brack{\Gamma\paren{\dfrac{e^{-\beta_0^*-\bm{X}_i^{*T}\boldsymbol{\beta}^*}}{1+e^{\beta_0}}}}\\ 
&= 0\\
&\geq \sup_{\beta_0} \log L(\boldsymbol{\theta}|\bm{W},\bm{M}),
\eal 
where the last inequality is because the log-likelihood associated with a discrete distribution cannot exceed 0.

Therefore,
$$ \sup_{\beta_0} \log L(\boldsymbol{\theta}|\bm{W},\bm{M}) = \lim_{\beta_0\to-\infty} \log L(\boldsymbol{\theta}|\bm{W},\bm{M})=0.$$

\end{proof}

\begin{thm}\label{theorem:2}
Consider the model \eqref{eq:bb1}--\eqref{eq:bb2} with parameters as in \eqref{eq:mu}--\eqref{eq:phi} and link functions as in  \eqref{eq:linkm}--\eqref{eq:linkp}.
Assume that $k=k^*=1$ and $X_i=X_i^* \in \set{0,1}$ for $i = 1,\ldots,n$. 
Suppose that $\sum_{i: X_i=0} W_i = 0$ and $\sum_{i: X_i=1} W_i > 0$.
Then the likelihood ratio test statistic for testing the null hypothesis that $\beta_1^*=0$ is equal to $0$.
\end{thm} 

\begin{proof}

Let $L_i$ represent the likelihood of the $i^{\text{th}}$ sample, so that 
\begin{equation} \label{eq:Li}
\sumi \log L_i(\beta_0,\beta_1,\beta_0^*,\beta_1^*|W_i, M_i) \equiv \log L(\beta_0,\beta_1,\beta_0^*,\beta_1^*|\bm{W},\bm{M}) .
\end{equation}

We wish to show that the likelihood ratio test statistic
\begin{align}
    \sup_{\beta_0, \beta_1, \beta_0^*, \beta_1^*} &\sumi \log L_i(\beta_0,\beta_1,\beta_0^*,\beta_1^*|W_i, M_i)\label{eq:wanttoshow} \\
    &\tab - \sup_{\beta_0,\beta_1,\beta_0^*} \sumi  L_i(\beta_0,\beta_1,\beta_0^*,\beta_1^*=0|W_i, M_i)=0\notag.
\end{align}

First, we notice that for all $i$, the parameters $\beta_0^*$ and $\beta_1^*$ enter the likelihood $L_i(\cdot)$ only though the term $\beta_0^*+\beta_1^*X_i$.
This term is equal to $\beta_0^*$ for all $i$ such that $X_i=0$, and $\beta_0^*+\beta_1^*$ for all $i$ such that $X_i=1$.
Similarly, the parameters $\beta_0$ and $\beta_1$ enter the likelihood $L_i(\cdot)$ only through the term $\beta_0$ for all $i$ such that $X_i=0$, and $\beta_0+\beta_1$ for all $i$ such that $X_i=1$.
Therefore, we can write the first term in \eqref{eq:wanttoshow} as the sum of two sub-problems,
\bal \sup_{\beta_0,\beta_1,\beta_0^*,\beta_1^*} &\sumi \log L_i(\beta_0,\beta_1,\beta_0^*,\beta_1^*| W_i,M_i)=\\
&\tab \sup_{\beta_0,\beta_1,\beta_0^*,\beta_1^*} \sum_{i:\ X_i = 0} \log L_i(\beta_0,\beta_1,\beta_0^*,\beta_1^*|W_i,M_i)\\
&\tab\tab + \sup_{\beta_0,\beta_1,\beta_0^*,\beta_1^*}\sum_{i:\ X_i=1} \log L_i(\beta_0,\beta_1,\beta_0^*,\beta_1^*|W_i,M_i)=\\
&\tab \sup_{\beta_0,\beta_1,\beta_0^*,\beta_1^*}\sum_{i:\ X_i=1} \log L_i(\beta_0,\beta_1,\beta_0^*,\beta_1^*|W_i,M_i),
\eal 
where the last equality results from Lemma~\ref{lemma:sub21}.
Similarly, we can write the second term in \eqref{eq:wanttoshow},
\bal 
\sup_{\beta_0,\beta_1,\beta_0^*} &\sumi \log L_i(\beta_0,\beta_1,\beta_0^*,\beta_1^*=0| W_i,M_i) = \sup_{\beta_0,\beta_1,\beta_0^*} \sum_{i:\ X_i=1} \log L_i(\beta_0,\beta_1,\beta_0^*,\beta_1^*=0|W_i,M_i).
\eal 
Thus, to show \eqref{eq:wanttoshow}, it suffices to show that
\begin{equation*}\label{eq:wanttoshow2}
\begin{split}
    \sup_{\beta_0, \beta_1, \beta_0^*, \beta_1^*} &\sum_{i:\ X_i=1} \log L_i(\beta_0, \beta_1, \beta_0^*, \beta_1^*|W_i,M_i) \\ 
    &\tab - \sup_{\beta_0,\beta_1,\beta_0^*} \sum_{i:\ X_i = 1} \log L_i(\beta_0, \beta_1, \beta_0^*, \beta_1^*=0|W_i,M_i) =0. 
    \end{split}
\end{equation*}
This follows directly from the fact that the parameters $\beta_0^*$ and $\beta_1^*$ enter the likelihood $L_i(\cdot)$ only though the term $\beta_0^*+\beta_1^*$ for all $i$ such that $X_i=1$. 
This completes our proof of Theorem~\ref{theorem:2}.
\end{proof}

\section*{Acknowledgements}
This work was partially supported by the NSF DGE-1256082 grant to Bryan D. Martin and the NSF CAREER DMS-1252624, NIH grant DP5OD009145, and a Simons Investigator Award in Mathematical Modeling of Living Systems grants to Daniela Witten.
The authors would thank Thea Whitman for providing data, as well as Pauline Trinh and Moira Differding for contributions to the software package.
The authors are also grateful to the R Core Team \citep{rcore} and authors of the packages brglm2 \citep{brglm}, ggplot2 \citep{ggplot}, phyloseq \citep{phyloseq}, trust \citep{trust}, and VGAM \citep{vgam}, which were used for constructing the figures and running the analyses in this article.

\end{document}